\documentclass[sigconf]{acmart}

\usepackage{graphicx}
\usepackage{amsmath,amsfonts}
\usepackage{mathtools}
\usepackage{bm}
\usepackage{subcaption}
\usepackage{booktabs}
\usepackage{multirow}
\usepackage{algorithm}
\usepackage{algpseudocode} 
\usepackage{tikz}
\usepackage{pgfplots}
\usepackage{xcolor}
\usepackage{balance}
\usepackage{placeins}
\usepackage{float}
\usepackage{enumerate}

\usepackage{caption}

\theoremstyle{plain}
\newtheorem{theorem}{Theorem}[section]
 
\newtheorem{proposition}[theorem]{Proposition}

\theoremstyle{definition}
\newtheorem{definition}[theorem]{Definition}
\newtheorem{example}{Example}
\theoremstyle{remark}
\newtheorem{remark}[theorem]{Remark}
\newtheorem{lemma}[theorem]{Lemma}

\begin{document}
\copyrightyear{2026}
\acmYear{2026}
\setcopyright{cc}
\setcctype{by-nc-nd}
\acmConference[HSCC '26]{29th ACM International Conference on Hybrid Systems: Computation and Control}{May 11--14, 2026}{Saint Malo, France}
\acmBooktitle{29th ACM International Conference on Hybrid Systems: Computation and Control (HSCC '26), May 11--14, 2026, Saint Malo, France}
\acmDOI{10.1145/3801146.3805668}
\acmISBN{979-8-4007-2566-1/2026/05}

\author{Peng Xie}
\affiliation{%
  \department{Department of Computer Engineering}
  \institution{TUM School of Computation, Information and Technology, Technical University of Munich}
  \city{Heilbronn}
  \country{Germany}}
\email{p.xie@tum.de}

\author{Amr Alanwar}
\affiliation{%
  \department{Department of Computer Engineering}
  \institution{TUM School of Computation, Information and Technology, Technical University of Munich}
  \city{Heilbronn}
  \country{Germany}}
\email{alanwar@tum.de}

\begin{abstract}
Data-driven reachability analysis using matrix zonotopes faces a fundamental challenge: the number of generators in the reachable set grows exponentially during propagation, while current order reduction yields overly conservative approximations in data-driven settings. This paper introduces an orthogonal matrix-based framework that appropriately transfers the coordinate system before reducing the generators of the reachable set, dramatically reducing reachable set volumes. By exploiting the factorized structure of data-driven matrix zonotope generators, we develop several efficient algorithms to solve the problem. Numerical experiments demonstrate order-of-magnitude volume reductions compared to traditional methods, while maintaining comparable generator numbers. Our method provides a practical solution to improve precision in data-driven safety verification.
\end{abstract}

\title{Orthogonal Transformations for Efficient Data-Driven Reachability Analysis}

\maketitle

\section{Introduction}

Data-driven reachability analysis verifies safety properties when precise system models are unavailable. Matrix zonotopes provide a computationally tractable representation for capturing model uncertainty from noisy trajectory data~\cite{alanwar2023data}: Minkowski sums and linear transformations admit closed-form expressions, enabling efficient set propagation~\cite{girard2005reachability,althoff2010reachability,stinson2016randomized} and widespread adoption in verification tools such as CORA~\cite{althoff2015cora}. Since Girard~\cite{girard2005reachability} established zonotope-based reachability for uncertain linear systems, the framework has been extended to hybrid systems via constrained zonotopes~\cite{scott2016constrained,xie2025data,althoff2007reachability}, mixed logical dynamical systems via hybrid zonotopes~\cite{bird2023hybrid}, and nonlinear dynamics via polynomial zonotopes~\cite{althoff2013reachability}. In the data-driven setting, Alanwar et al.~\cite{alanwar2023data} introduced matrix zonotopes to encapsulate parametric uncertainty, with recent refinements through constrained matrix zonotopes~\cite{alanwar2023robust}.

Despite these theoretical advances, a fundamental computational challenge remains. While the generator count increases linearly with time in model-based settings~\cite{girard2005reachability}, data-driven approaches face far more severe growth. Matrix zonotopes inherently require numerous generators to capture model uncertainty, and these generators compound during propagation. Furthermore, computing exact zonotope volumes requires evaluating determinants over all $n$-subsets of generators---a combinatorial problem that becomes intractable for moderate dimensions~\cite{gover2010determinants,esterov2010zonotopes,guibas2003zonotopes,emiris2014efficient}.

This computational burden makes order reduction essential. Girard~\cite{girard2005reachability} aggregates smaller generators into axis-aligned interval hulls, maintaining bounded generator counts with outer-approximation guarantees. Subsequent refinements include PCA-based dominant-direction methods~\cite{yang2011methods}, bounding-box quality analysis~\cite{dimitrov2000bounds}, rotational invariance~\cite{ding2019rotational}, SVD-oriented bounding boxes for hybrid systems~\cite{stursberg2003efficient}, and recent comparative studies~\cite{kopetzki2017methods,chen2024novel}. In contrast to~\cite{stursberg2003efficient}, which orients the hull representation itself, the present work applies an orthogonal coordinate transformation before Girard reduction to a \emph{data-driven matrix zonotope}, exploiting its rank-one generator structure arising from the noise model. The goal is not a generic reduction scheme for arbitrary zonotopes, but to leverage the matrix structure specific to data-driven reachability for tighter outer approximations.

Without order reduction, exponential generator growth makes long-horizon analysis prohibitive~\cite{girard2005reachability}; yet standard reduction in the original coordinates produces overly conservative volumes~\cite{alanwar2023data}. This trade-off has limited the practical applicability of data-driven reachability for safety-critical systems.

\begin{figure*}[!t]
\centering
\includegraphics[width=0.95\textwidth]{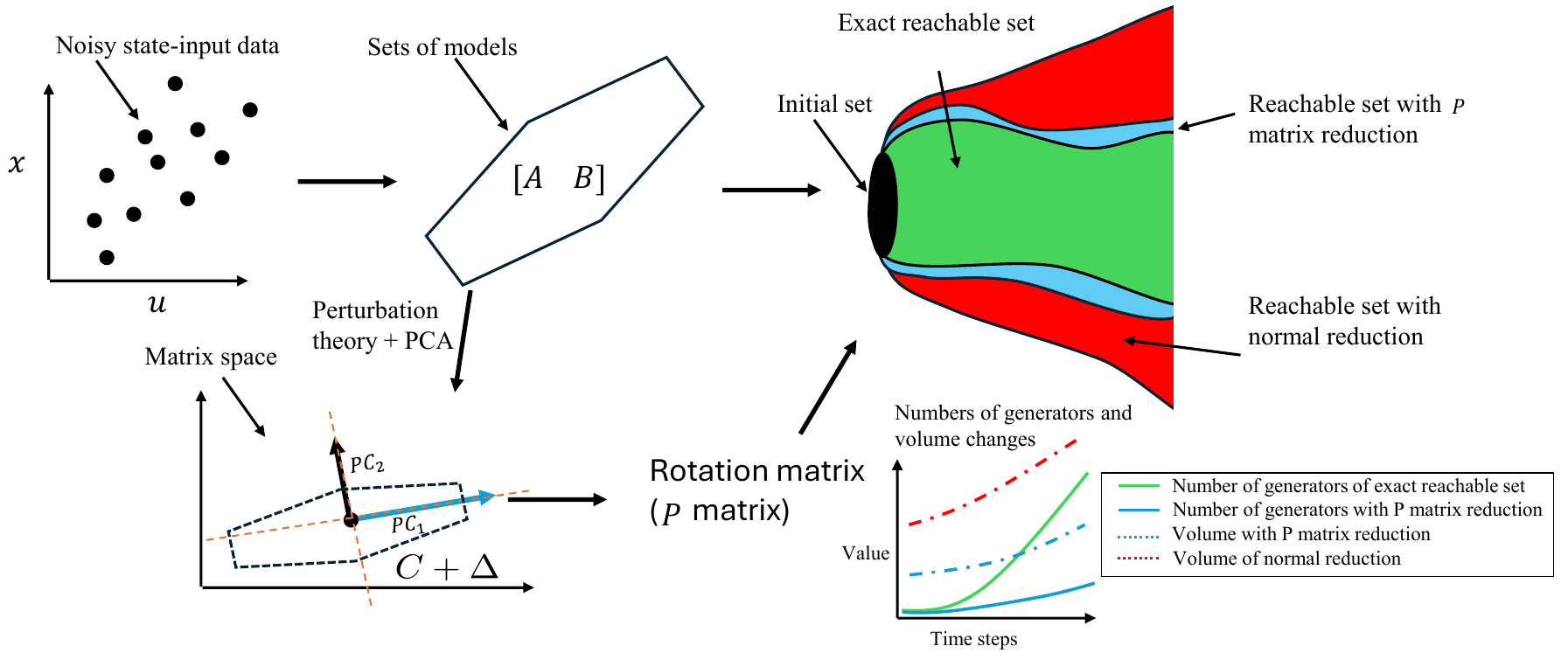}
\caption{Overview of the proposed approach. Noisy data yield a matrix zonotope of feasible dynamics. An orthogonal matrix~$P$, selected via perturbation theory and PCA, aligns coordinates with dominant directions, producing tighter reachable sets (blue) than standard reduction (red). The exact set (green) is the ground truth.}
\label{fig:main_workflow}
\end{figure*}

The proposed solution, illustrated in Figure~\ref{fig:main_workflow}, leverages a key insight: the effectiveness of order reduction depends fundamentally on the coordinate system. By choosing an orthogonal matrix~$P$ that aligns coordinates with dominant dynamical directions, the Girard reduction~\cite{girard2005reachability} operates more efficiently, minimizing the wrapping effect. As the figure shows, projection-based reduction (blue) maintains volumes orders of magnitude smaller than standard reduction (red) at comparable generator counts.

Projection selection is formulated as an optimization on orthogonal manifolds, drawing on classical perturbation theory (Davis--Kahan~\cite{davis1970rotation}, Wedin~\cite{wedin1972perturbation}, Stewart--Sun~\cite{stewart1990matrix}) to quantify how noise affects invariant subspaces~\cite{knyazev2002principal}, and on manifold optimization~\cite{edelman1998geometry,absil2009optimization} via Givens-parameterized coordinate descent~\cite{shalit2014coordinate} and Manopt solvers~\cite{boumal2014manopt,gao2021riemannian}.

This paper makes three contributions. First, we analyze reachable set propagation through matrix zonotopes, establishing connections between perturbations, orientation selection, and volume after Girard reduction. Second, we develop a projection optimization framework that minimizes reduced zonotope volumes under a \emph{fixed generator budget}; the contribution is improved tightness per reduction step (mitigating the wrapping effect), not a change in the intrinsic generator-growth rate. Third, we present numerical evaluations demonstrating substantial volume reductions at comparable generator counts; the full code will be made available upon acceptance.

The remainder of the paper is organized as follows. Section~\ref{sec-prelim} reviews set representations and perturbation theory. Section~\ref{Invariance} develops the orthogonal transformation framework and its volume analysis. Section~\ref{p_matrix_selection} presents algorithms for selecting the orthogonal matrix. Section~\ref{sec:numerical} provides numerical evaluations, and Section~\ref{conclu} concludes.


\section{PRELIMINARIES AND PROBLEM STATEMENT}\label{sec-prelim}

\subsection{Notation}
Matrices are denoted by uppercase Roman letters (e.g., $A \in \mathbb{R}^{n \times n}$, $P \in \mathbb{R}^{n \times r}$), whereas calligraphic capitals represent sets (e.g., $\mathcal{R}_k \subset \mathbb{R}^n$, $\mathcal{M}_{\Sigma}$). Vectors and scalars appear as lowercase Roman letters (e.g., $x_k \in \mathbb{R}^n$, $\alpha_i \in [-1,1]$). The symbols $\mathbf{1}_n$ and $\mathbf{0}_n$ denote the all-ones and all-zeros vectors in $\mathbb{R}^n$, respectively, while $I_n$ denotes the $n \times n$ identity matrix. The transpose is indicated by $^{\top}$. For a matrix $A$, the $(i,j)$-th entry is denoted $A_{(i,j)}$ and the $j$-th column is $A_{(:,j)}$. $\mathcal{O}(\cdot)$ denotes big-O notation: $a_k = \mathcal{O}(b_k)$ means $a_k$ is bounded by a constant multiple of $b_k$ for all sufficiently large $k$.

The spectral norm of a matrix $M$ is denoted $\|M\|_2$, and its Frobenius norm is $\|M\|_F$. The singular values of $M \in \mathbb{R}^{m \times n}$ are denoted $\sigma_1(M) \geq \sigma_2(M) \geq \cdots \geq \sigma_{\min(m,n)}(M)$, with $\sigma_{\min}(M)$ representing the smallest singular value. For symmetric matrices $A$, eigenvalues are denoted $\lambda_1(A) \geq \lambda_2(A) \geq \cdots \geq \lambda_n(A)$. The notation $\sin\Theta(U_1, U_2)$ represents the matrix of canonical angles between the subspaces spanned by the columns of $U_1$ and $U_2$. The Minkowski sum of sets $\mathcal{A}$ and $\mathcal{B}$ is denoted $\mathcal{A} \oplus \mathcal{B} = \{a + b : a \in \mathcal{A}, b \in \mathcal{B}\}$. The Cartesian product is denoted $\mathcal{A} \times \mathcal{B}$. $\mathcal{A} \cap \mathcal{B}$ denotes their set-theoretic intersection. $\mathrm{vol}(\cdot)$ denotes the Lebesgue volume of a set. $\operatorname{GR}(\cdot)$ denotes the Girard zonotope order-reduction operator in~\cite{girard2005reachability}.


\subsection{Set Representations}
The reachability pipeline relies on the following set representations.

\begin{definition}[Zonotope~\cite{girard2005reachability}]
$\mathcal{Z} = \langle c, G \rangle = \{c + G\xi : \xi \in [-1,1]^p\}$, where $c \in \mathbb{R}^n$ and $G \in \mathbb{R}^{n \times p}$.
\end{definition}

\begin{definition}[Matrix zonotope {\cite[p.~54]{althoff2010reachability}}]
\label{def:matrix_zonotope}
Given $C, G^{(i)} \in \mathbb{R}^{n \times m}$,
\begin{equation}
\mathcal{M} = \bigl\langle C, \{G^{(i)}\}_{i=1}^{\gamma} \bigr\rangle
= \Bigl\{C + \textstyle\sum_{i=1}^{\gamma} \beta_i G^{(i)} : \beta_i \in [-1,1]\Bigr\}.
\label{eq:matrix_zonotope}
\end{equation}
\end{definition}

\begin{definition}[Constrained Zonotope \cite{scott2016constrained}]
\label{def:constrained_zonotope}
Given $G \in \mathbb{R}^{n \times n_g}$, $A \in \mathbb{R}^{n_c \times n_g}$, $b \in \mathbb{R}^{n_c}$,
\[
\mathcal{CZ} = \langle G, c, A, b \rangle = \{c + G\xi : \xi \in [-1,1]^{n_g},\, A\xi = b\}.
\]
Constrained zonotopes arise in the intersection-based tightening of Lemma~\ref{thm:intersection_refinement}.
\end{definition}

\begin{proposition}[Operations on Constrained Zonotopes \cite{scott2016constrained}]
\label{prop:constrained_operations}
For constrained zonotopes $\mathcal{CZ} = \langle G_z, c_z, A_z, b_z \rangle \subset \mathbb{R}^n$, $\mathcal{W} = \langle G_w, c_w, A_w,$ $ b_w \rangle \subset \mathbb{R}^n$, $\mathcal{Y} = \langle G_y, c_y, A_y, b_y \rangle \subset \mathbb{R}^k$, and matrix $R \in \mathbb{R}^{k \times n}$, the following identities hold:
\begin{align}
R\mathcal{CZ} &= \langle RG_z, Rc_z, A_z, b_z \rangle, \\
\mathcal{CZ} + \mathcal{W} &= \left\langle [G_z \; G_w], c_z + c_w, \begin{bmatrix} A_z & 0 \\ 0 & A_w \end{bmatrix}, \begin{bmatrix} b_z \\ b_w \end{bmatrix} \right\rangle, \\
\mathcal{CZ} \cap_R \mathcal{Y} &= \left\langle [G_z \; 0], c_z, \begin{bmatrix} A_z & 0 \\ 0 & A_y \\ RG_z & -G_y \end{bmatrix}, \begin{bmatrix} b_z \\ b_y \\ c_y - Rc_z \end{bmatrix} \right\rangle.
\end{align}
\end{proposition}

\begin{proposition}[{Multiplication between a Matrix Zonotope and a Constrained Zonotope in~\cite[Proposition~2]{alanwar2023data}}]
\label{prop:matzonotope_conzonotope}
Given a matrix zonotope, $\mathcal{M} = \langle C, \{G^{(1)},$ $ \ldots, G^{(\gamma)}\} \rangle$ with $C \in \mathbb{R}^{n \times m}$, $G^{(i)} \in \mathbb{R}^{n \times m}$, and a constrained zonotope $\mathcal{CZ} = \langle G_z, c_z, A_z, b_z \rangle$ with $G_z \in \mathbb{R}^{m \times n_g}$, $c_z \in \mathbb{R}^m$, the set
\[
\mathcal{M} \times \mathcal{CZ} \subseteq \bar{\mathcal{CZ}}
= \langle \bar{c}, \bar{G}, \bar{A}, \bar{b} \rangle
\]

with
\begin{align}
\bar{c} &= C c_z, \\
\bar{G} &= \bigl[\, C G_z,\; G^{(1)} c_z,\dots, G^{(\gamma)} c_z,\; G^{(1)} G_z,\dots, G^{(\gamma)} G_z \,\bigr], \\
\bar{A} &= \mathrm{blkdiag}\Bigl( A_z,\; I_{\gamma},\; \underbrace{A_z,\dots,A_z}_{\gamma} \Bigr), \\
\bar{b} &= \bigl[\, b_z^{\top},\; 0^{\top},\; \underbrace{b_z^{\top},\dots,b_z^{\top}}_{\gamma} \bigr]^{\top}.
\end{align}
\end{proposition}

\subsubsection{Spectral Perturbation Theory}
The following classical results quantify how perturbations affect singular and eigensubspaces.


\begin{theorem}[Wedin's Perturbation Theorem~\cite{wedin1972perturbation}]
Let $M, \tilde{M} = M + E \in \mathbb{R}^{m \times n}$ have singular value decompositions
\[
M = U \Sigma V^{\top}, \qquad \tilde{M} = \tilde{U} \tilde{\Sigma} \tilde{V}^{\top}.
\]
Fix an index $r$ with $1 \le r \le \min\{m,n\}$, and write
\[
U = [\,U_r \; U_\perp\,], \quad V = [\,V_r \; V_\perp\,],
\]
where $U_r \in \mathbb{R}^{m \times r}$ and $V_r \in \mathbb{R}^{n \times r}$ contain the left and right singular vectors associated with the $r$ largest singular values $\sigma_1(M) \ge \cdots \ge \sigma_r(M)$. Define $\delta := \min\{\sigma_r(M) - \sigma_{r+1}(M),\, \sigma_r(M)\} > 0$. Then
\begin{equation}
\max \bigl\{ \|\sin \Theta(U_r, \tilde{U}_r)\|_2,\; \|\sin \Theta(V_r, \tilde{V}_r)\|_2 \bigr\}
\;\le\; \frac{\|E\|_2}{\delta},
\end{equation}
where $\tilde{U}_r$ and $\tilde{V}_r$ are defined analogously from $\tilde{M}$.
\end{theorem}

\begin{theorem}[Davis--Kahan $\sin\Theta$ Theorem \cite{davis1970rotation}]
Let $A, \tilde{A} = A + E \in \mathbb{R}^{n \times n}$ be symmetric matrices with eigendecompositions $A = U\Lambda U^{\top}$ and $\tilde{A} = \tilde{U}\tilde{\Lambda}\tilde{U}^{\top}$. Let $\mathcal{S} \subset \{1, \ldots, n\}$ be an index set with spectral gap $\delta = \min_{i \in \mathcal{S}, j \notin \mathcal{S}} |\lambda_i(A) - \lambda_j(A)| > 0$. Then:
\begin{equation}
\|\sin\Theta(U_{\mathcal{S}}, \tilde{U}_{\mathcal{S}})\|_2 \leq \frac{\|E\|_2}{\delta},
\end{equation}
where $U_{\mathcal{S}}$ contains the eigenvectors corresponding to eigenvalues indexed by $\mathcal{S}$.
\end{theorem}

\subsection{Problem Statement}
With the above tools in place, we formulate the data-driven reachability problem.

Consider the discrete-time linear system
\begin{equation}
x_{k+1} = A x_k + B u_k + w_k, \qquad w_k \in \mathcal{Z}_w \subset \mathbb{R}^n,
\label{eq:sys}
\end{equation}
where $x_k \in \mathbb{R}^n$ and $u_k \in \mathbb{R}^m$ denote the state and input, respectively, $\mathcal{Z}_w$ is a known zonotopic bound on the disturbance, and the system matrices $A \in \mathbb{R}^{n \times n}$ and $B \in \mathbb{R}^{n \times m}$ are unknown.

We are given input--state trajectory data from $K$ experiments of lengths $T_i+1$:
\[
\{x^{(i)}(k)\}_{k=0}^{T_i}, \quad \{u^{(i)}(k)\}_{k=0}^{T_i-1}, \quad i = 1,\dots,K.
\]
By stacking all trajectories, we obtain
\begin{align*}
X_{+} &= \bigl[ x^{(1)}(1) \;\cdots\; x^{(1)}(T_1) \;\cdots \\
&\qquad\quad x^{(K)}(1) \;\cdots\; x^{(K)}(T_K) \bigr], \\
X_{-} &= \bigl[ x^{(1)}(0) \;\cdots\; x^{(1)}(T_1{-}1) \;\cdots \\
&\qquad\quad x^{(K)}(0) \;\cdots\; x^{(K)}(T_K{-}1) \bigr], \\
U_{-} &= \bigl[ u^{(1)}(0) \;\cdots\; u^{(1)}(T_1{-}1) \;\cdots \\
&\qquad\quad u^{(K)}(0) \;\cdots\; u^{(K)}(T_K{-}1) \bigr],
\end{align*}
and let $T = \sum_{i=1}^K T_i$ be the total number of data points.

Following the data-driven reachability formulation in~\cite{alanwar2023data}, all matrices $[A\;B]$ consistent with the data and bounded noise can be over-approximated by the matrix zonotope
\begin{equation}
\mathcal{M}_{\Sigma}
= \bigl( X_{+} - \mathcal{M}_w \bigr)
\begin{bmatrix}
X_{-} \\
U_{-}
\end{bmatrix}^{\dagger},
\label{eq:msigma}
\end{equation}
where $\mathcal{M}_w$ accounts for the disturbance. Theorem~\ref{thm:noise_to_matZono} later shows that $\mathcal{M}_{\Sigma}$ is concretely a matrix zonotope
\[
\mathcal{M}_{\Sigma} = \bigl\langle C,\, \{G^{(i,j)}\}_{i=1,\dots,p}^{j=1,\dots,T} \bigr\rangle,
\]
with center $C = X_{+}H$ and rank-one generators $G^{(i,j)} = -g_w^{(i)} h_j^{\top}$, where $H$ is any right inverse of $[X_{-}; U_{-}]$, $\{g_w^{(i)}\}$ are the noise generators of $\mathcal{M}_w$, and $h_j^{\top}$ is the $j$-th row of $H$.
The corresponding data-driven reachable sets are propagated as
\begin{equation}
\hat{\mathcal{R}}_{k+1}
= \mathcal{M}_{\Sigma} \bigl( \hat{\mathcal{R}}_k \times \mathcal{U}_k \bigr) \oplus \mathcal{Z}_w.
\label{eq:dd-reach}
\end{equation}

A central difficulty with the recursion \eqref{eq:dd-reach} is the rapid growth of the zonotopic representation. Let
\[
\mathcal{M}_{\Sigma} = \langle C, \{ G^{(1)}, \dots, G^{(\gamma_M)} \} \rangle
\]
denote the data-driven matrix zonotope with $\gamma_M$ matrix generators, let $\hat{\mathcal{R}}_k$ have $g_k$ generators, let the input set $\mathcal{U}_k$ have $g_u$ generators, and let the noise set $\mathcal{Z}_w$ have $g_w$ generators. Using the standard over-approximation for the product of a matrix zonotope with a constrained zonotope, the next reachable set \(\hat{\mathcal{R}}_{k+1}\) contains
\begin{equation}
g_{k+1}
= (1 + \gamma_M)\bigl( g_k + g_u \bigr) + \gamma_M + g_w
\label{eq:gen-growth-exact}
\end{equation}
generators. Since $\gamma_M \ge 1$ whenever the data implies any model uncertainty, the factor $(1+\gamma_M) > 1$ leads to effectively exponential growth of $g_k$ with $k$; equivalently,
\[
g_k = (1+\gamma_M)^k g_0 + \mathcal{O}\bigl( (1+\gamma_M)^k \bigr).
\]
This generator explosion directly translates into rapidly expanding reachable sets and prohibitive computational effort. The goal of this work is to modify the data-driven reachability construction so as to retain the model consistency of \eqref{eq:msigma}--\eqref{eq:dd-reach} while substantially mitigating the generator growth described in \eqref{eq:gen-growth-exact}.

\section{Girard Order Reduction With Orthogonal Transformation}
\label{Invariance}

Since the generator count grows exponentially with the time horizon~\eqref{eq:gen-growth-exact}, order reduction is unavoidable. Girard's reduction~\cite{girard2005reachability} uses axis-aligned interval hulls to aggregate discarded generators, making the resulting volume dependent on the coordinate system. Our work improves the \emph{tightness} of this reduction under a fixed generator budget by choosing an orthogonal matrix~$P$ that rotates into a frame where axis-aligned aggregation is most efficient, thereby reducing the wrapping effect. We first establish a basic invariance property.
\begin{theorem}[Data-Driven Reachability under Orthogonal Transformations]
\label{thm:orthogonal_invariance}
Let $P \in \mathbb{R}^{n \times n}$ be an orthogonal matrix, i.e., $P^{\top}P = PP^{\top} = I$, and consider the coordinate change $z_k = P^{\top} x_k$. Let \eqref{eq:msigma} be the data-driven matrix zonotope containing all $[A\;B]$ consistent with the data and bounded noise~\cite{alanwar2023data}. Then:
\begin{enumerate}[(i)]
\item The transformed data-driven set of models is
\begin{equation}
\tilde{\mathcal{M}}_{\Sigma}
= P^{\top} \mathcal{M}_{\Sigma}
\begin{bmatrix}
P & 0 \\
0 & I_m
\end{bmatrix},
\label{eq:msigma-transformed}
\end{equation}
i.e., it is exactly the original matrix zonotope seen in the rotated coordinates.
\item If the reachable sets are propagated as
\begin{equation}
\hat{\mathcal{R}}_{k+1}
= \mathcal{M}_{\Sigma}(\hat{\mathcal{R}}_k \times \mathcal{U}_k) \oplus \mathcal{Z}_w,
\end{equation}
then the reachable sets in the transformed coordinates,
\begin{equation}
\tilde{\mathcal{R}}_{k+1}
= \tilde{\mathcal{M}}_{\Sigma}(\tilde{\mathcal{R}}_k \times \mathcal{U}_k) \oplus \tilde{\mathcal{Z}}_w,
\qquad
\tilde{\mathcal{Z}}_w = P^{\top} \mathcal{Z}_w,
\end{equation}
satisfy
\[
\tilde{\mathcal{R}}_{k+1} = P^{\top} \hat{\mathcal{R}}_{k+1}
\quad \text{for all } k \ge 0,
\]
provided $\tilde{\mathcal{R}}_0 = P^{\top} \hat{\mathcal{R}}_0$.
\end{enumerate}
In other words, rotating the state space by an orthogonal matrix, doing data-driven reachability there, and rotating back yields exactly the same sets.
\end{theorem}

\begin{proof}
Under $z = P^{\top} x$, the data matrices transform as
\[
\tilde{X}_{+} = P^{\top} X_{+}, \qquad
\tilde{X}_{-} = P^{\top} X_{-}, \qquad
\tilde{\mathcal{M}}_{w} = P^{\top} \mathcal{M}_{w}.
\]
Therefore, the data-driven set of models in the new coordinates is
\begin{equation}
\label{eq:tilde-MSigma-start}
\tilde{\mathcal{M}}_{\Sigma}
= (\tilde{X}_{+} - \tilde{\mathcal{M}}_{w})
\begin{bmatrix}
\tilde{X}_{-} \\
U_-
\end{bmatrix}^{\dagger}
= (P^{\top} X_{+} - P^{\top} \mathcal{M}_{w})
\begin{bmatrix}
P^{\top} X_{-} \\
U_-
\end{bmatrix}^{\dagger}.
\end{equation}
Now use the Moore–Penrose pseudo-inverse identity for orthogonal matrices: for
\[
\begin{bmatrix}
P^{\top} & 0 \\
0 & I_m
\end{bmatrix}
\begin{bmatrix}
X_- \\
U_-
\end{bmatrix}
=
\begin{bmatrix}
P^{\top} X_- \\
U_-
\end{bmatrix},
\]
we have
\begin{equation}
\label{eq:pinv-identity}
\begin{bmatrix}
P^{\top} X_- \\
U_-
\end{bmatrix}^{\dagger}
=
\begin{bmatrix}
X_- \\
U_-
\end{bmatrix}^{\dagger}
\begin{bmatrix}
P & 0 \\
0 & I_m
\end{bmatrix},
\end{equation}
substituting \eqref{eq:pinv-identity} into \eqref{eq:tilde-MSigma-start} yields
\begin{equation}
\begin{aligned}
\tilde{\mathcal{M}}_{\Sigma}
&= P^{\top} (X_{+} - \mathcal{M}_{w})
\begin{bmatrix}
X_- \\
U_-
\end{bmatrix}^{\dagger}
\begin{bmatrix}
P & 0 \\
0 & I_m
\end{bmatrix} \\
&= P^{\top} \mathcal{M}_{\Sigma}
\begin{bmatrix}
P & 0 \\
0 & I_m
\end{bmatrix},
\end{aligned}
\end{equation}
which is exactly \eqref{eq:msigma-transformed}.

For the reachable sets, assume $\tilde{\mathcal{R}}_k = P^{\top} \hat{\mathcal{R}}_k$. Then
\begin{align*}
\tilde{\mathcal{R}}_{k+1}
&= \tilde{\mathcal{M}}_{\Sigma}(\tilde{\mathcal{R}}_k \times \mathcal{U}_k) \oplus \tilde{\mathcal{Z}}_w \\
&= P^{\top} \mathcal{M}_{\Sigma}
\begin{bmatrix}
P & 0 \\
0 & I_m
\end{bmatrix}
\begin{bmatrix}
P^{\top} & 0 \\
0 & I_m
\end{bmatrix}
(\hat{\mathcal{R}}_k \times \mathcal{U}_k) \oplus P^{\top} \mathcal{Z}_w \\
&= P^{\top} \mathcal{M}_{\Sigma}(\hat{\mathcal{R}}_k \times \mathcal{U}_k) \oplus P^{\top} \mathcal{Z}_w \\
&= P^{\top} \hat{\mathcal{R}}_{k+1}.
\end{align*}

By induction over $k$, the claim holds for all $k \ge 0$.
\end{proof}

\subsection{Integration with Girard Order Reduction}

While orthogonal transformations preserve data-driven reachable sets exactly, their practical benefit appears when combined with order-reduction techniques. We recall the Girard reduction operator $\operatorname{GR}(\cdot)$~\cite{girard2005reachability}: given a zonotope $\mathcal{Z} = \langle c, G \rangle$ with $G \in \mathbb{R}^{n \times p}$ and a target order~$\rho$ (so the reduced zonotope has at most $\rho n$ generators), $\operatorname{GR}(\mathcal{Z})$ keeps the $\rho n - n$ generators with largest $\ell_2$ norm and replaces the remaining generators by the smallest axis-aligned box (interval hull) that contains them. The interval hull half-widths are
$d_i = \sum_{j \in \mathcal{I}_{\text{disc}}} |G_{ij}|$,
and the result is
\[
\operatorname{GR}(\mathcal{Z}) = \bigl\langle c,\; [G_{:,\mathcal{I}_{\text{keep}}},\; \mathrm{diag}(d)] \bigr\rangle,
\]
which is an outer approximation: $\mathcal{Z} \subseteq \operatorname{GR}(\mathcal{Z})$. Because the aggregation is axis-aligned, the quality of the reduction depends heavily on the coordinate system, motivating our orthogonal transformation approach.

Let $\hat{\mathcal{R}}_{k+1} \subset \mathbb{R}^n$ be a zonotopic reachable set and let $P \in \mathbb{R}^{n \times n}$ be orthogonal. We consider the composite operation
\begin{enumerate}
\item \textit{projection}: $P^{\top} \hat{\mathcal{R}}_{k+1}$,
\item \textit{order reduction}: apply Girard’s reduction $\operatorname{GR}(\cdot)$ to $P^{\top} \hat{\mathcal{R}}_{k+1}$,
\item \textit{back-mapping}: map the reduced set back via $P$.
\end{enumerate}
The resulting set is
\[
P \cdot \operatorname{GR}\bigl(P^{\top} \hat{\mathcal{R}}_{k+1}\bigr).
\]
Now consider the data-driven reachable set
\begin{equation}
\hat{\mathcal{R}}_{k+1}
= \mathcal{M}_{\Sigma}(\hat{\mathcal{R}}_k \times \mathcal{U}_k) \oplus \mathcal{Z}_w
= \langle c_{k+1}, G_{k+1} \rangle,
\end{equation}
with center $c_{k+1} \in \mathbb{R}^n$ and generator matrix $G_{k+1} \in \mathbb{R}^{n \times p}$. Directly applying Girard reduction to $\hat{\mathcal{R}}_{k+1}$ in the original coordinates tends to be conservative, because many small generators are aggregated along fixed axes. By first rotating with a suitable $P$ (e.g., aligned with dominant directions of $\mathcal{M}_{\Sigma}$), these generators become more structured in the transformed space, and the reduction step discards less relevant directions more effectively. After back-mapping we obtain
\begin{equation}
P \cdot \operatorname{GR}\bigl(P^{\top} \hat{\mathcal{R}}_{k+1}\bigr)
= \langle P \tilde{c}_{k+1}^{\mathrm{red}},\; P \tilde{G}_{k+1}^{\mathrm{red}} \rangle,
\end{equation}
which is an outer approximation of $\hat{\mathcal{R}}_{k+1}$ but typically tighter than applying $\operatorname{GR}(\cdot)$ directly in the original coordinates. For a well-chosen $P$ we observe
\begin{equation}\label{eq:vol_comp}
\mathrm{vol}\bigl( P \cdot \operatorname{GR}(P^{\top} \hat{\mathcal{R}}_{k+1}) \bigr)
\ll
\mathrm{vol}\bigl( \operatorname{GR}(\hat{\mathcal{R}}_{k+1}) \bigr).
\end{equation}

\begin{definition}[Dominant directions]
\label{def:dominant_directions}
Let $\hat{\mathcal{R}}_{k+1} = \langle c, G \rangle$ and let $G_{\text{disc}} \in \mathbb{R}^{n \times m}$ denote the submatrix of generators that would be discarded under a fixed reduction order. Define
$S := G_{\text{disc}} G_{\text{disc}}^{\top} \in \mathbb{R}^{n \times n}$.
The \emph{dominant directions} are the eigenvectors of $S$ (equivalently, the left singular vectors of $G_{\text{disc}}$) corresponding to its largest eigenvalues. A \emph{PCA basis} of $G_{\text{disc}}$ is any orthogonal $P$ whose columns diagonalize $S$, i.e., $P^{\top}SP$ is diagonal.
\end{definition}

\begin{lemma}[Containment and volume reduction]
\label{rem:containment_volume}
Let $P \in \mathbb{R}^{n \times n}$ be orthogonal and let $\hat{\mathcal{R}}_{k+1}$ be a zonotope. Then
\begin{enumerate}
\item (\emph{Containment})
\[
\hat{\mathcal{R}}_{k+1}
\subseteq
P \cdot \operatorname{GR}\bigl(P^{\top} \hat{\mathcal{R}}_{k+1}\bigr).
\]
\item (\emph{Volume reduction}) A PCA basis of the discarded generators (Definition~\ref{def:dominant_directions}) minimizes a computable upper bound on the volume of the interval-hull term injected by Girard reduction, and in practice yields significantly smaller volumes as shown in Example~\ref{ex:volume_reduction}.
\end{enumerate}
\end{lemma}

\begin{proof}
\textit{Containment:} Let $x \in \hat{\mathcal{R}}_{k+1}$ be arbitrary. Since $P$ is orthogonal, $x = P P^{\top} x$. Define $z = P^{\top} x$, then $z \in P^{\top} \hat{\mathcal{R}}_{k+1}$. By the outer-approximation property of Girard reduction \cite[Theorem~2]{girard2005reachability},
\begin{equation}
P^{\top} \hat{\mathcal{R}}_{k+1}
\subseteq
\operatorname{GR}(P^{\top} \hat{\mathcal{R}}_{k+1}).
\end{equation}
Hence $z \in \operatorname{GR}(P^{\top} \hat{\mathcal{R}}_{k+1})$, and after back-mapping we obtain
\begin{equation}
x = P z \in P \cdot \operatorname{GR}(P^{\top} \hat{\mathcal{R}}_{k+1}),
\end{equation}
which proves containment.

\textit{Volume reduction:}
Let $G_{\text{disc}} \in \mathbb{R}^{n \times m}$ be the submatrix of discarded generators (the $m$ generators with smallest $\ell_2$ norms). The interval hull of the discarded part has half-width vector
\[
d(P) = |P^{\top}G_{\text{disc}}| \mathbf{1}_m,
\qquad
d_i(P) = \|e_i^{\top} P^{\top} G_{\text{disc}}\|_1,
\]
and the volume of the injected box $\mathcal{B}(d(P))$ satisfies
\begin{equation}\label{eq:box_vol_bound}
\mathrm{vol}(\mathcal{B}(d(P))) = 2^n \prod_{i=1}^n d_i(P)
\;\le\; (2\sqrt{m})^n \sqrt{\prod_{i=1}^n (P^{\top}SP)_{ii}},
\end{equation}
where $S = G_{\text{disc}}G_{\text{disc}}^{\top}$.
The inequality follows from the Cauchy--Schwarz inequality:
\[
d_i(P) \le \sqrt{m}\,\|e_i^{\top}P^{\top}G_{\text{disc}}\|_2
= \sqrt{m\,(P^{\top}SP)_{ii}}.
\]

The right-hand side of \eqref{eq:box_vol_bound} is minimized over all orthogonal $P$ by any $P$ that diagonalizes $S$, i.e., a PCA basis of $G_{\text{disc}}$. This follows from the Hadamard inequality: for any positive semidefinite matrix $Q$,
$\det(Q) \le \prod_i Q_{ii}$,
with equality if and only if $Q$ is diagonal. Since $\det(P^{\top}SP) = \det(S)$ for all orthogonal $P$, minimizing $\prod_i (P^{\top}SP)_{ii}$ is equivalent to making $P^{\top}SP$ diagonal, which is achieved by the eigenvectors of $S$.

We note that \eqref{eq:box_vol_bound} provides a rigorous motivation for aligning $P$ with dominant directions, as it minimizes a computable upper bound on the volume of the interval hull injected by Girard reduction. However, it does not imply a universal ordering between
$\mathrm{vol}(\operatorname{GR}(\hat{\mathcal{R}}_{k+1}))$
and
$\mathrm{vol}(P\operatorname{GR}(P^{\top}\hat{\mathcal{R}}_{k+1}))$
for arbitrary~$P$, because the total reduced-set volume also depends on the kept generators and Minkowski-sum interactions. The strong volume improvements are therefore presented as an empirical observation, confirmed in Example~\ref{ex:volume_reduction} (a 244-fold reduction) and Section~\ref{sec:numerical}.
\end{proof}

We can further tighten the over-approximation by combining the two reduced sets through intersection. 

\begin{lemma}[Intersection-Based Refinement]
\label{thm:intersection_refinement}
Let $\hat{\mathcal{R}}_{k+1}$ be a zonotope and $P \in \mathbb{R}^{n \times n}$ be an orthogonal matrix. The intersection of direct and projection-based reductions yields a constrained zonotope that contains the original set:
\begin{equation}
\hat{\mathcal{R}}_{k+1} \subseteq \mathcal{CZ}_{\text{int}} = \operatorname{GR}(\hat{\mathcal{R}}_{k+1}) \cap P \cdot \operatorname{GR}(P^{\top}\hat{\mathcal{R}}_{k+1})
\end{equation}
where $\mathcal{CZ}_{\text{int}}$ is a constrained zonotope computed using Proposition~\ref{prop:constrained_operations}.
\end{lemma}

\begin{proof}
The containment follows from the fundamental properties of set operations. Since Girard reduction preserves over-approximation (\cite[Theorem 2]{girard2005reachability}), we have:
\begin{equation}
\hat{\mathcal{R}}_{k+1} \subseteq \operatorname{GR}(\hat{\mathcal{R}}_{k+1})
\end{equation}
From Lemma~\ref{rem:containment_volume}, we also established:
\begin{equation}
\hat{\mathcal{R}}_{k+1} \subseteq P \cdot \operatorname{GR}(P^{\top}\hat{\mathcal{R}}_{k+1})
\end{equation}
Since $\hat{\mathcal{R}}_{k+1}$ is contained in both sets, it must be contained in their intersection:
\begin{equation}
\hat{\mathcal{R}}_{k+1} \subseteq \operatorname{GR}(\hat{\mathcal{R}}_{k+1}) \cap P \cdot \operatorname{GR}(P^{\top}\hat{\mathcal{R}}_{k+1})
\end{equation}
\end{proof}

\begin{example}
\label{ex:volume_reduction}
Consider a five-dimensional system with
\begin{align*}
\mathcal{Z}_w &= \langle 0,\; 0.06 \cdot \mathrm{diag}([1, 2, 1.3, 1, 1.5]) \rangle, \\
\mathcal{X}_0 &= \langle 3 \cdot \mathbf{1}_5,\; G_0 \rangle, \\
\mathcal{U}_0 &= \langle 10,\; 0.25 \rangle,
\end{align*}
where $G_0 \in \mathbb{R}^{5 \times 6}$ contains the prescribed generators. After one data-driven propagation step with a fixed orthogonal matrix $P$, the volumes are
\begin{itemize}
\item direct Girard reduction: $\mathrm{vol}(\operatorname{GR}(\hat{\mathcal{R}}_1)) = 1438$,
\item projection-based reduction: $\mathrm{vol}(P \cdot \operatorname{GR}(P^{\top}\hat{\mathcal{R}}_1)) = 5.88$.
\end{itemize}
Thus, the projection-based reduction yields about a 244-fold volume decrease while still remaining an outer approximation.
\end{example}

In Figure~\ref{fig:reduction_comparison_updated}, the green solid boundary (a constrained zonotope constructed as in Proposition~\ref{prop:constrained_operations}) depicts the intersection of the two reduced sets and gives the tightest outer approximation among them. This set is still not smaller than the original reachable set (blue dashed line). 


For the projection $P$ in Example~\ref{ex:volume_reduction},
\begin{align*}
\mathrm{vol}(\hat{\mathcal{R}}_{k+1}) &< \mathrm{vol}(\mathcal{CZ}_{\text{int}})
< \mathrm{vol}\!\bigl(P \cdot \operatorname{GR}(P^{\top}\hat{\mathcal{R}}_{k+1})\bigr) \\
&= 5.88
\ll \mathrm{vol}\!\bigl(\operatorname{GR}(\hat{\mathcal{R}}_{k+1})\bigr) = 1438.
\end{align*}
When constrained zonotopes are affordable, $\mathcal{CZ}_{\text{int}}$ is the preferred approximation; otherwise the projection-based reduction alone offers a substantial gain. The approach applies to any representation closed under orthogonal maps.

\begin{figure}
\centering
\includegraphics[width=\columnwidth]{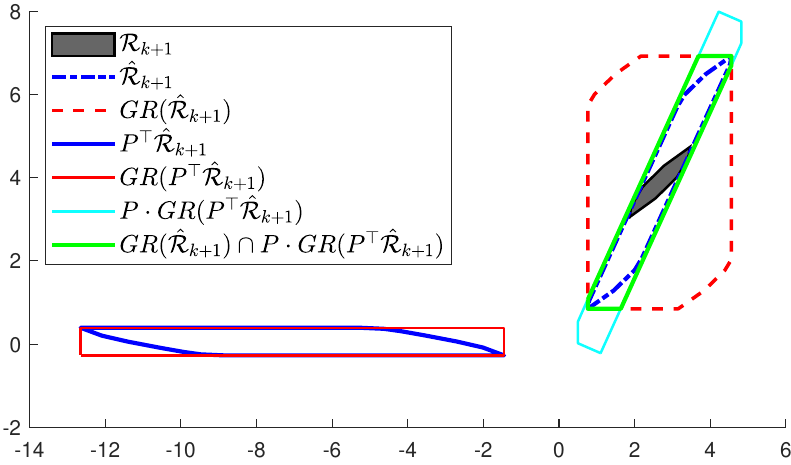}
\caption{Comparison of order reduction strategies with intersection refinement. The plot shows projections onto the 2D plane of: model-based ground truth (gray filled), original data-driven set $\hat{\mathcal{R}}_{k+1}$ (blue dashed), direct Girard reduction $\operatorname{GR}(\hat{\mathcal{R}}_{k+1})$ (red dashed), transformed set $P^{\top}\hat{\mathcal{R}}_{k+1}$ (blue solid), transformed reduction $\operatorname{GR}(P^{\top}\hat{\mathcal{R}}_{k+1})$ (red solid), back-projected set $P \cdot \operatorname{GR}(P^{\top}\hat{\mathcal{R}}_{k+1})$ (cyan), and the intersection $\operatorname{GR}(\hat{\mathcal{R}}_{k+1}) \cap P \cdot \operatorname{GR}(P^{\top}\hat{\mathcal{R}}_{k+1})$ (green solid).}
\label{fig:reduction_comparison_updated}
\end{figure}

\subsection{Optimal Orthogonal Matrix Selection}
Having established the projection--reduction framework, we now formulate the selection of~$P$ as an optimization problem.

We seek the orthogonal matrix minimizing the reachable-set volume on the Stiefel manifold $\mathcal{V}_{n,n} = \{P \in \mathbb{R}^{n \times n} : P^{\top}P = I_n\}$:
\begin{equation}
(P^{\star}, J^{\star}) = \left(\arg\min_{P \in \mathcal{V}_{n,n}} J(P),\ \min_{P \in \mathcal{V}_{n,n}} J(P)\right),
\end{equation}
where $J(P) = \operatorname{vol}(\mathcal{X}_P^{(T)})$ is evaluated via Algorithm~\ref{alg:projection_propagation_refined}.


\subsection{Algorithms for Orthogonal Matrix  Optimization}

We present the data-driven forward projection-based propagation in Algorithm~\ref{alg:projection_propagation_refined} and two complementary approaches for solving the projection optimization problem: an exhaustive multi-scale coordinate descent~\cite{parra2021rotation, shalit2014coordinate} in Algorithm~\ref{alg:givens_multiscale} and an efficient Riemannian optimization method~\cite{gao2021riemannian} in Algorithm~\ref{alg:riemannian_opt}. In Algorithm~\ref{alg:projection_propagation_refined}, the operator $\otimes_*$ denotes either standard or constrained multiplication depending on whether $\mathcal{X}^{(t-1)}$ is a zonotope or a constrained zonotope.
\begin{algorithm}[h]
\caption{Projection-Based Propagation with Intersection}
\label{alg:projection_propagation_refined}
\begin{algorithmic}[1]
\Require Orthogonal matrix $P \in \mathbb{R}^{n \times n}$, matrix zonotope $\mathcal{M}_{\Sigma}$, initial set $\mathcal{X}_0$ , input zonotope $\mathcal{U}$, noise zonotope $\mathcal{Z}_w$, horizon $T$, order $\rho$, flag \texttt{use\_intersection}
\Ensure Reachable set $\mathcal{X}^{(T)}$ 
\State Initialize: $\mathcal{M}_P = P^{\top}\mathcal{M}_{\Sigma}P$, $\mathcal{X}_P^{(0)} = P^{\top}\mathcal{X}_0$
\For{$t = 1$ to $T$}
    \State $\mathcal{Z}_P^{(t)} = \mathcal{M}_P \otimes_c \mathcal{X}_P^{(t-1)} \oplus \mathcal{Z}_{w,P}$ \Comment{Use Prop.~\ref{prop:matzonotope_conzonotope}}
    
    \State $\tilde{\mathcal{X}}_P^{(t)} = \operatorname{GR}(\mathcal{Z}_P^{(t)}, \rho)$

    \If{\texttt{use\_intersection}}
        \State $\mathcal{CZ}^{(t)} = \mathcal{M}_{\Sigma} \otimes_* \mathcal{X}^{(t-1)} \oplus \mathcal{Z}_w$
        \State $\tilde{\mathcal{X}}^{(t)} = \operatorname{GR}(\mathcal{CZ}^{(t)}, \rho)$
        \State $\mathcal{X}^{(t)} = \tilde{\mathcal{X}}^{(t)} \cap (P \cdot \tilde{\mathcal{X}}_P^{(t)})$ \Comment{Constrained zonotope}
        \State $\mathcal{X}_P^{(t)} = P^{\top}\mathcal{X}^{(t)}$
    \Else
        \State $\mathcal{X}_P^{(t)} = \tilde{\mathcal{X}}_P^{(t)}$
    \EndIf
\EndFor
\State \Return $ \mathcal{X}^{(T)} = P \cdot \mathcal{X}_P^{(T)}$
\end{algorithmic}
\end{algorithm}


\subsubsection{Multi-Scale Givens Coordinate Descent}

The Givens coordinate-descent approach explores the orthogonal group through 2D rotations, updating only one coordinate pair at a time and preserving orthogonality by construction; see, e.g., \cite{shalit2014coordinate,parra2021rotation} for related schemes based on Givens rotations. Any $n \times n$ orthogonal matrix can be written as a product of Givens (plane) rotations~\cite{shalit2014coordinate}, so iterating over all coordinate pairs and all rotation angles on a discrete grid provides a complete parameterization of the orthogonal group. At each scale we scan all coordinate pairs and all angles on a discrete grid, accept a rotation only if it reduces the objective, and then refine the grid at the next (smaller) scale.

\begin{algorithm}[!t]
\caption{Multi-Scale Givens Coordinate Descent}
\label{alg:givens_multiscale}
\begin{algorithmic}[1]
\Require Initial projection $P_0 \in \mathcal{V}_{n,n}$; scale levels $s_j = \pi \cdot 2^{-j}$ for $j = 0, 1, \dots, J$; improvement threshold $\epsilon > 0$
\Ensure Optimized projection $P^{\star}$
\State $P = P_0$, $J_{\text{current}} = \operatorname{vol}(\mathcal{X}_P^{(T)})$ \label{line:givens_init}
\For{$j = 0$ to $J$} \label{line:scale_loop}
    \State $\text{converged} = \text{False}$
    \State Define angle grid $\Theta_j = \{-\pi, -\pi + s_j, \dots, \pi\}$ \label{line:grid_define}
    \While{not converged} \label{line:convergence_loop}
        \State $\text{improved} = \text{False}$
        \For{all pairs $(i_1,i_2)$ with $1 \le i_1 < i_2 \le n$} \label{line:pair_loop}
            \State $\theta^{\star} = 0$, $J_{\text{best}} = J_{\text{current}}$
            \For{each $\theta \in \Theta_j$} \label{line:angle_search}
                \State Construct Givens rotation $G(i_1,i_2,\theta)$:
                \State \quad $G = I_n + (\cos\theta - 1)(e_{i_1}e_{i_1}^\top + e_{i_2}e_{i_2}^\top)
                + \sin\theta (e_{i_1}e_{i_2}^\top - e_{i_2}e_{i_1}^\top)$ \label{line:givens_construct}
                \State $P_{\text{trial}} = G P$ \label{line:trial_projection}
                \State $J_{\text{trial}} = \operatorname{vol}(\mathcal{X}_{P_{\text{trial}}}^{(T)})$ \Comment{via Alg.~\ref{alg:projection_propagation_refined}} \label{line:volume_eval}
                \If{$J_{\text{trial}} < J_{\text{best}}$} \label{line:update_best}
                    \State $J_{\text{best}} = J_{\text{trial}}$, $\theta^{\star} = \theta$
                \EndIf
            \EndFor
            \If{$\bigl(J_{\text{current}} - J_{\text{best}}\bigr)/J_{\text{current}} > \epsilon$} \label{line:improvement_check}
                \State Rebuild $G^\star = G(i_1,i_2,\theta^{\star})$
                \State $P = G^\star P$
                \State $J_{\text{current}} = J_{\text{best}}$
                \State $\text{improved} = \text{True}$
            \EndIf
        \EndFor
        \State $\text{converged} = \neg \text{improved}$ \label{line:convergence_check}
    \EndWhile
\EndFor
\State \Return $P$
\end{algorithmic}
\end{algorithm}

\begin{algorithm}[!t]
\caption{Multi-Scale Givens Descent with Orthogonal Symmetry}
\label{alg:givens_multiscale_full}
\begin{algorithmic}[1]
\Require Initial projection $P_0 \in \mathcal{V}_{n,n}$; scale levels $\{s_j\}_{j=0}^J$; threshold $\epsilon > 0$
\State $P^{(1)} \gets$ Alg.~\ref{alg:givens_multiscale}$(P_0, \{s_j\}, \epsilon)$
\State $P^{(2)} \gets$ Alg.~\ref{alg:givens_multiscale}$(-P_0, \{s_j\}, \epsilon)$
\State $J^{(1)} \gets \operatorname{vol}(\mathcal{X}_{P^{(1)}}^{(T)})$, $J^{(2)} \gets \operatorname{vol}(\mathcal{X}_{P^{(2)}}^{(T)})$
\If{$J^{(1)} \le J^{(2)}$}
    \State \Return $P^{(1)}$
\Else
    \State \Return $P^{(2)}$
\EndIf
\end{algorithmic}
\end{algorithm}

Algorithm~\ref{alg:givens_multiscale} implements a coarse-to-fine search over decreasing angle steps $s_j = \pi 2^{-j}$ (line~\ref{line:scale_loop}). At each scale it sweeps all $\binom{n}{2}$ coordinate pairs (line~\ref{line:pair_loop}), testing Givens rotations on the grid $\Theta_j$ (line~\ref{line:angle_search}) and accepting updates that exceed the threshold~$\epsilon$ (line~\ref{line:improvement_check}). The cost is $O(J\, n^2\, |\Theta_j|\, C_{\text{vol}})$, where $C_{\text{vol}}$ (volume evaluation via reachability) dominates.

\begin{remark}[Note on orthogonality.]
Algorithm~\ref{alg:givens_multiscale} uses only Givens rotations, so all iterates remain in the special orthogonal group $SO(n)$ (i.e., orthogonal matrices with $\det=1$) \cite{shalit2014coordinate, frerix2019approximating}. The full orthogonal group $O(n)$ has two connected components, distinguished by $\det=+1$ and $\det=-1$, and the reduction may benefit from either component. To cover both, the Givens search is run twice: once from $P_0$ and once from a reflected seed $S P_0$ with $\det S = -1$ (e.g., $S=\mathrm{diag}(-1,1,\dots,1)$), retaining the candidate with the smaller volume. This incurs only a factor of~2 in cost while exploring all orthogonal symmetries relevant to the reduction.

\end{remark}

\subsubsection{Riemannian Optimization on the Stiefel Manifold}

To address the computational limitations of exhaustive search, we employ Riemannian optimization that directly exploits the geometric structure of the Stiefel manifold~\cite{gao2021riemannian,tagare2011notes}. This approach performs efficient gradient-based optimization while preserving orthogonality through intrinsic geometric operations.

\begin{algorithm}[!t]
\caption{Riemannian Optimization for Projection Matrix}
\label{alg:riemannian_opt}
\begin{algorithmic}[1]
\Require Matrix zonotope $\mathcal{M}_{\Sigma}$, initial projections $\{P_1, \ldots, P_R\} \subset \mathcal{V}_{n,n}$, convergence tolerance $\varepsilon > 0$, maximum iterations $I_{\max} \in \mathbb{N}$, maximum trust-region radius $\Delta_{\max} > 0$
\Ensure Projection $P^{\star} \in \mathcal{V}_{n,n}$ minimizing volume, and $J^{\star}$
\State Define manifold $\mathcal{M} := \mathcal{V}_{n,n} = \{P \in \mathbb{R}^{n \times n} : P^{\top}P = I_n\}$ \label{line:manifold_def}
\State Define objective $f(P) := \log(\operatorname{vol}(\mathcal{X}_P^{(T)}))$ \Comment{$\mathcal{X}_P^{(T)}$ from Alg.~\ref{alg:projection_propagation_refined}} \label{line:objective_def}
\State $J^{\star} = +\infty$, $P^{\star} = P_1$ \label{line:riem_init}
\For{$i = 1$ to $R$} \label{line:multistart_loop}
    \State $P^{(0)} = P_i$, $\Delta^{(0)} = 1$, $\ell = 0$
    \While{$\ell < I_{\max}$} \label{line:riem_main_loop}
        \State Compute Riemannian gradient $\mathrm{grad}_f(P^{(\ell)}) \in T_{P^{(\ell)}}\mathcal{M}$ \label{line:riem_gradient}
        \If{$\|\mathrm{grad}_f(P^{(\ell)})\| < \varepsilon$} \label{line:convergence_test}
            \State \textbf{break}
        \EndIf
        \State Solve trust-region subproblem on $T_{P^{(\ell)}}\mathcal{M}$ to get $\xi^{(\ell)}$ \label{line:trust_subproblem}
        \State $P_{\text{trial}} = \mathcal{R}_{P^{(\ell)}}(\xi^{(\ell)})$ \label{line:retraction}
        \State Evaluate ratio $\rho^{(\ell)}$ \label{line:ratio_eval}
        \If{$\rho^{(\ell)} > 0.25$}
            \State $P^{(\ell+1)} = P_{\text{trial}}$
            \If{$\rho^{(\ell)} > 0.75$}
                \State $\Delta^{(\ell+1)} = \min(2\Delta^{(\ell)}, \Delta_{\max})$
            \EndIf
        \Else
            \State $P^{(\ell+1)} = P^{(\ell)}$
            \State $\Delta^{(\ell+1)} = 0.25 \Delta^{(\ell)}$ \label{line:shrink_trust}
        \EndIf
        \State $\ell = \ell + 1$
    \EndWhile
    \State $P_{\text{local}} = P^{(\ell)}$, $J_{\text{local}} = \exp(f(P_{\text{local}}))$ \label{line:extract_local}
    \If{$J_{\text{local}} < J^{\star}$} \label{line:update_global}
        \State $P^{\star} = P_{\text{local}}$, $J^{\star} = J_{\text{local}}$
    \EndIf
\EndFor
\State \Return $(P^{\star}, J^{\star})$
\end{algorithmic}
\end{algorithm}

Algorithm~\ref{alg:riemannian_opt} optimizes $f(P) = \log(\operatorname{vol}(\mathcal{X}_P^{(T)}))$ directly on the Stiefel manifold via trust-region steps, so all iterates remain orthogonal. A multi-start procedure (line~\ref{line:multistart_loop}) from diverse initializations (data-based, PCA-based, random QR) mitigates local minima. Each run computes the Riemannian gradient by projecting the Euclidean gradient onto $T_P\mathcal{V}_{n,n}$ (line~\ref{line:riem_gradient}), solves a trust-region subproblem (line~\ref{line:trust_subproblem}), and retracts onto the manifold via $\mathcal{R}_P(\xi) = (P + \xi)(I + \xi^\top \xi)^{-1/2}$ (line~\ref{line:retraction}). The best solution over all starts is returned (line~\ref{line:update_global}).

\begin{remark}[Initialization strategy]
The quality of the multi-start initializations significantly affects convergence speed and solution quality. In practice, including the heuristic projections developed in Section~\ref{p_matrix_selection} (e.g., $\mathcal{L}_1$-SVD, maximal rotation) alongside random orthogonal matrices ensures that at least one starting point lies near a good local minimum, typically reducing the number of trust-region iterations required.
\end{remark}

\section{Special Orthogonal  Matrices}\label{p_matrix_selection}

The rotation-based order reduction described above can, in principle, find good projections, but it is computationally demanding. To reduce this cost, we exploit the structure of data-driven matrix zonotopes together with spectral perturbation theory to construct effective orthogonal matrices efficiently.
Before presenting the algorithms, we establish the fundamental connection between noise in data-driven reachability analysis and the resulting matrix zonotope structure.



\begin{algorithm}[!t]
\caption{Unified Orthogonal Matrix Selection Framework}\label{alg:unified_projection}
\begin{algorithmic}[1]
\Require Matrix zonotope $\mathcal{M}_{\Sigma}=\langle C,\{G^{(1)},\dots,G^{(\gamma)}\}\rangle$, dimension $n$
\Ensure $P\in\mathbb{R}^{n\times n}$ with $P^\top P=I_n$
\If{method = $\mathcal{L}_1$-SVD} 
    \State $P=$ Algorithm~\ref{alg:svd_pca}($\mathcal{M}_{\Sigma}$, $n$)
\ElsIf{method = Max-Rotation} 
    \State $P=$ Algorithm~\ref{alg:max_rotation_pca}($\mathcal{M}_{\Sigma}$, $n$)
\EndIf
\State \Return $P$
\end{algorithmic}
\end{algorithm}

\begin{theorem}[Noise to Matrix Zonotope Factorization]
\label{thm:noise_to_matZono}
Consider the data equation
\begin{equation}\label{eq:model}
X_+ \;=\; AX_- + BU_- + W_- \;=\; \Sigma\, M + W_-,
\end{equation}
where $\Sigma=[A\ B]\in\mathbb{R}^{n\times(n+m)}$ and $M=\begin{bmatrix}X_-\\ U_-\end{bmatrix}\in\mathbb{R}^{(n+m)\times T}$.
Assume the following:
\begin{enumerate}
\item For each $j\in\{1,\dots,T\}$, the $j$-th noise column admits
\begin{equation}\label{eq:wj}
w_j \;=\; \sum_{i=1}^{p}\xi_{i,j}\,g_w^{(i)},\qquad \xi_{i,j}\in[-1,1]
\end{equation}
with the coefficients $\{\xi_{i,j}\}$ free and independent across $(i,j)$.
\item $M$ has full row rank $n{+}m$. Fix any right inverse
\begin{equation}\label{eq:rightinv}
H\in\mathbb{R}^{T\times(n+m)}\quad\text{s.t.}\quad M H = I_{n+m}.
\end{equation}
\end{enumerate}
Then any $\Sigma$ consistent with \eqref{eq:model} must satisfy
\begin{equation}\label{eq:sigma_affine}
\Sigma \;=\; X_+H - W_-H,
\end{equation}
and hence belongs to the matrix zonotope
\begin{equation}\label{eq:matzono}
\Sigma \;\in\; \Big\langle\, C,\ \{G^{(i,j)}\}_{\substack{i=1,\ldots,p \\ j=1,\ldots,T}} \,\Big\rangle,
\end{equation}
with center $C=X_+H$ and rank-one generators
\begin{equation}\label{eq:gen_rankone}
G^{(i,j)}=-\,g_w^{(i)}\,h_j^{\top},
\end{equation}
where $h_j^{\top}$ is the $j$-th row of $H$.

\end{theorem}

\begin{proof}
Right-multiplying \eqref{eq:model} by $H$ and using \eqref{eq:rightinv} gives \eqref{eq:sigma_affine}. 
\begin{gather}
W_-=\sum_{j=1}^T w_j e_j^{\top},\quad
w_j=\sum_{i=1}^{p}\xi_{i,j}g_w^{(i)}, \label{eq:WH}\\
W_-H=\sum_{j=1}^{T}\sum_{i=1}^{p}\xi_{i,j}\, g_w^{(i)} h_j^{\top}, \notag
\end{gather}
and substituting into \eqref{eq:sigma_affine} yields \eqref{eq:matzono}.
\end{proof}

\begin{remark}[Representation and rank-one structure]\label{rmk:rightfactor}
Each generator in \eqref{eq:matzono} has rank at most $1$:
$G^{(i,j)}=-g_w^{(i)}h_j^{\top}$. 

\end{remark}

Theorem~\ref{thm:noise_to_matZono} reveals that all generators share the common right factor~$H$. We now present two algorithms that exploit this structure.




\subsection{$\mathcal{L}_1$-SVD Projection}

Algorithm~\ref{alg:svd_pca} aggregates generators using $\ell_1$ norms to capture the overall uncertainty spread, then extracts principal directions via eigendecomposition.

\begin{figure*}[!t]
    \centering
\includegraphics[width=0.95\textwidth]{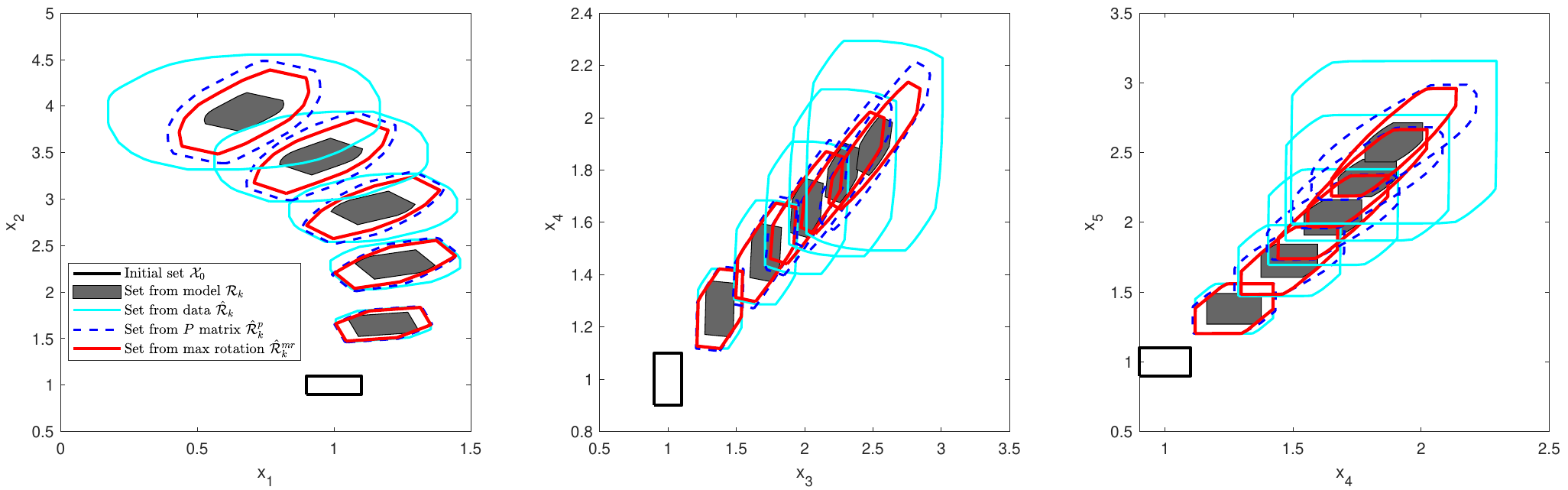}
    \caption{2D projections in 5D state space: ground truth $\mathcal{R}_k$ (gray regions), initial set $\mathcal{X}_0$ (black box), data-driven reachable set $\hat{\mathcal{R}}_k$ (cyan line), set obtained with matrix $P$ (blue dashed line), and set obtained with the maximal-rotation method $\hat{\mathcal{R}}^{mr}_k$ (red solid line). The maximal-rotation method achieves a significantly tighter outer approximation.} \label{fig:matrix-zonotope-projections}
\end{figure*}

\subsection{Maximal Rotation Projection}
We now exploit Theorem~\ref{thm:noise_to_matZono} to construct an orthogonal matrix whose principal directions are most sensitive to admissible matrix-zonotope perturbations. Let
\begin{gather}
S_0 = \tfrac{1}{n+m}\, C C^{\top}= U\Lambda U^{\top}, \label{eq:S0}\\
\Lambda=\mathrm{diag}(\lambda_1,\ldots,\lambda_n),\quad \lambda_1\ge\cdots\ge\lambda_n, \notag
\end{gather}
and fix a target eigen-index $i$. Denote the nearest spectral competitor:
\begin{equation}\label{eq:gap}
j^{\star}=\arg\min_{j\neq i}|\lambda_i-\lambda_j|,\qquad
\delta=|\lambda_i-\lambda_{j^{\star}}|>0.
\end{equation}
Consider admissible perturbations (coefficients bounded by one)
\begin{equation}\label{eq:Ebeta}
E(\beta)=-\sum_{r=1}^{p}\sum_{t=1}^{T}\beta_{r,t}\, g_w^{(r)} h_t^{\top},\quad
|\beta_{r,t}|\le 1,
\end{equation}
and the perturbed covariance
\begin{equation}\label{eq:Sbeta}
S(\beta)=\frac{(C+E(\beta))(C+E(\beta))^{\top}}{n+m},\quad
\Delta S(\beta)=S(\beta)-S_0.
\end{equation}
By the Davis--Kahan $\sin\Theta$ theorem~\cite{bhardwaj2024matrix},
\begin{equation}\label{eq:DK}
\sin\theta_i(\beta)\ \le\ \frac{\|\Delta S(\beta)\|_2}{\delta}.
\end{equation}
Using the rank-one structure in \eqref{eq:Ebeta},
\begin{equation}\label{eq:DeltaS_bound}
\|\Delta S(\beta)\|_2 \ \le\ \frac{2}{n+m}\sum_{r,t} |\beta_{r,t}|\,\|C h_t\|_2\,\|g_w^{(r)}\|_2
+\frac{\|E(\beta)\|_2^2}{n+m},
\end{equation}
where
\begin{equation}\label{eq:E_norm_bound}
\|E(\beta)\|_2 \ \le\ \sum_{r,t} |\beta_{r,t}|\,\|g_w^{(r)}\|_2\,\|h_t\|_2.
\end{equation}
Inequalities \eqref{eq:DK}–\eqref{eq:E_norm_bound} yield a rigorous a priori bound on the maximal rotation of $u_i$ under admissible matrix-zonotope perturbations. For small-amplitude noise, one may drop the quadratic term in \eqref{eq:DeltaS_bound}.

\begin{algorithm}[!t]
\caption{$\mathcal{L}_1$-SVD Projection}\label{alg:svd_pca}
\begin{algorithmic}[1]
\Require Matrix zonotope $\mathcal{M}_{\Sigma} = \langle C, \{G^{(1)}, \ldots, G^{(\gamma)}\} \rangle$, dimension $n$
\Ensure Orthogonal projection $P \in \mathbb{R}^{n \times n}$
\State Initialize $G \in \mathbb{R}^{n \times \gamma}$
\For{$i = 1$ to $\gamma$}
    \State $G_{:,i} = \sum_{j=1}^{n+m} |G^{(i)}_{:,j}|$ \Comment{Row-wise $\ell_1$ aggregation}
\EndFor
\State $X = [G, -G]^{\top} \in \mathbb{R}^{2\gamma \times n}$
\State $S = \frac{1}{2\gamma}X^{\top}X$; $S = V\Lambda V^{\top}$ 
\State \Return $P = V$ sorted by eigenvalues
\end{algorithmic}
\end{algorithm}

\begin{algorithm}[t]
\caption{Maximal-Rotation Projection (zonotope input)}
\label{alg:max_rotation_pca}
\begin{algorithmic}[1]
\Require Matrix zonotope $\mathcal{M}_{\Sigma} = \langle C, \{G^{(1)}, \ldots, G^{(\gamma)}\} \rangle$, dimension $n$
\Ensure $P\in\mathbb{R}^{n\times n}$ orthogonal ($P^{\top}P=I_n$)
\State $d=$ number of columns of $C$;
\State $S_0 = \frac{1}{d}CC^{\top}$; $S_0=U\Lambda U^{\top}$
\State Set target eigen-index $i= 1$ (dominant); \State $j^{\star}= \arg\min_{j\neq i}|\lambda_i-\lambda_j|$;  $\delta= |\lambda_i-\lambda_{j^{\star}}|$
\State $u_i= U(:,i)$; \ $u_{j^{\star}}= U(:,j^{\star})$ \label{ui_uj}
\For{$k=1$ to $\gamma$} \Comment{first-order coupling for generator $k$}
  \State $c_k = (u_{j^{\star}}^{\top} C G^{(k)\top} u_i)\;+\; (u_{j^{\star}}^{\top} G^{(k)} C^{\top} u_i)$
\EndFor
\State $\beta_k^{\star}= \mathrm{sign}(c_k)$ for all $k$ \ (with $\mathrm{sign}(0)=0$)
\State $E^{\star}= -\sum_{k=1}^{\gamma}\beta_k^{\star}\,G^{(k)}$;$C^{\star}= C+E^{\star}$
\State $S^{\star}= \frac{1}{d}C^{\star}C^{\star\top}$;$S^{\star}=V\Sigma_{\star}V^{\top}$;$P= V$
\State \Return $P$
\end{algorithmic}
\end{algorithm}

Algorithm~\ref{alg:max_rotation_pca} linearizes $\Delta S(\beta)$ from~\eqref{eq:Sbeta} around~$C$ and projects the variation onto the two most competitive eigendirections $(u_i,u_{j^{\star}})$ (line~\ref{ui_uj}), yielding a first-order coupling score per generator. Maximizing this coupling selects a vertex of the zonotope that approximately maximizes the rotation of~$u_i$; the output basis is then obtained from the eigendecomposition of $S^{\star}=\tfrac{1}{d}C^{\star}C^{\star\top}$, where $d$ is the number of columns of~$C$.

Figure~\ref{fig:matrix-zonotope-projections} illustrates the effect: starting from a random $5\times 5$ orthogonal~$P$, Algorithm~\ref{alg:max_rotation_pca} produces an improved matrix~$P_{mr}$ whose reachable set $\hat{\mathcal{R}}^{mr}_k$ is visibly tighter than the baseline~$\hat{\mathcal{R}}_k$ at negligible extra cost.

The effectiveness of these methods follows from the rank-1 structure of generators in Theorem~\ref{thm:noise_to_matZono}.

\begin{remark}[Degenerate cases]
\label{rem:degenerate}
When the eigenvalue spectrum of $S = G_{\text{disc}}G_{\text{disc}}^{\top}$ is flat ($\lambda_1(S)/\lambda_n(S)\approx 1$), the bound~\eqref{eq:box_vol_bound} is nearly the same for every orthogonal~$P$, and the gain over the baseline is modest. This situation arises when the noise affects all state dimensions equally, leaving no dominant direction to align with. Conversely, when the spectrum is well-separated, the projection can align the coordinate axes with the low-variance subspace, concentrating the reduced generators where they matter most. In all cases the outer-approximation guarantee is preserved, so applying the method is never harmful.
\end{remark}



\section{Numerical Examples}\label{sec:numerical}

\begin{figure*}[h]
\centering
\includegraphics[width=0.95\textwidth]{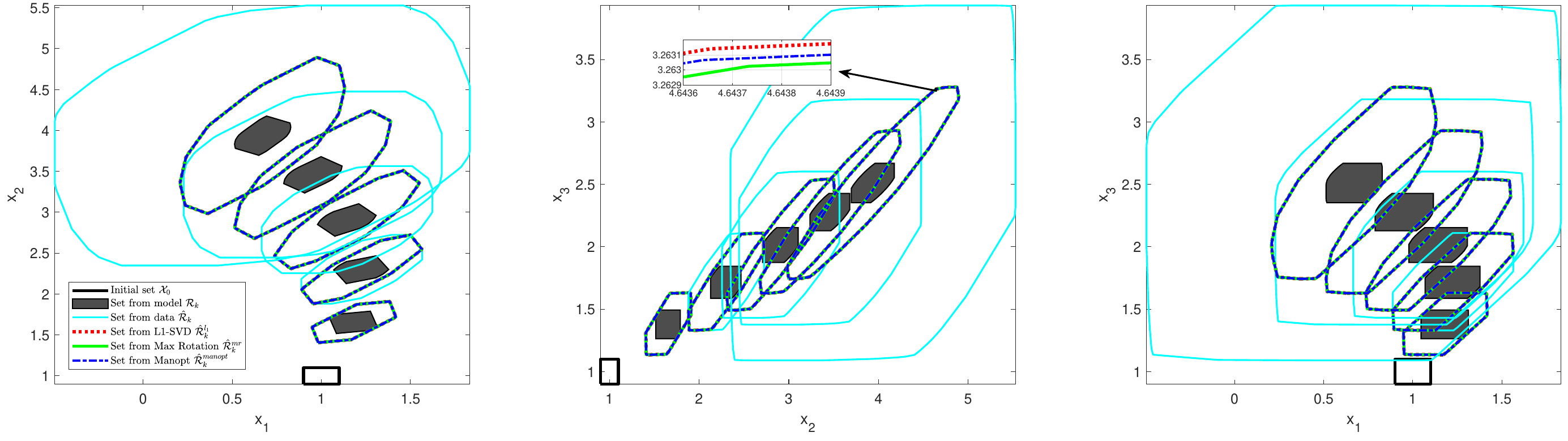}
\caption{Reachable set approximations obtained using different projection methods under process noise. Ground truth $\mathcal{R}_k$ (gray regions), initial set $\mathcal{X}_0$ (black box), data-driven reachable set $\hat{\mathcal{R}}_k$ (cyan line). The inset highlights the minimal differences between $\mathcal{L}_1$-SVD $\hat{\mathcal{R}}^{\mathcal{L}_1}$ (red dashed line), maximal rotation $\hat{\mathcal{R}}^{mr}$ (green solid line), and the Manopt method $\hat{\mathcal{R}}^{manopt}$ (blue dash-dot line).}
\label{fig:all_methods_comparison}
\end{figure*}

This section demonstrates the effectiveness of the proposed projection matrix methods through numerical simulations.

Consider a five-dimensional discrete-time linear system with initial state set $\mathcal{X}_0 = \langle \mathbf{1}_5, 0.1 I_5 \rangle$, scalar control input bounded by $\mathcal{U} = \langle 10, 0.25 \rangle$, and process noise $\mathcal{W} = \langle \mathbf{0}_5, 0.01\times\text{diag}([1, 1.1, 1.3, 1.2, 1.5]) \rangle$. The system matrices are:

\begin{equation*}
A = \begin{bmatrix}
0.9323 & -0.189 & 0 & 0 & 0 \\
0.189 & 0.9323 & 0 & 0 & 0 \\
0 & 0 & 0.8596 & 0.04302 & 0 \\
0 & 0 & -0.04302 & 0.8596 & 0 \\
0 & 0 & 0 & 0 & 0.9048
\end{bmatrix},
\end{equation*}
\begin{equation*}
B^T = \begin{bmatrix}
0.04363 ,0.05327, 0.04754, 0.04528,0.04758
\end{bmatrix}.
\end{equation*}

Figure~\ref{fig:all_methods_comparison} visualizes the reachable-set approximations under different projections; the inset shows that the advanced methods are nearly indistinguishable at plotting resolution. Table~\ref{tab:projection_performance} reports runtimes and projected volumes (computed in CORA~\cite{althoff2015cora}). The two reachable sets
$\hat{\mathcal{R}}^{\mathcal{L}_1}_k$ and $\hat{\mathcal{R}}^{mr}_k$,
obtained by $\mathcal{L}_1$--SVD and maximal rotation respectively, achieve virtually identical accuracy (volume ratios $\approx 7.33$--$7.36$) while being orders of magnitude faster than the Riemannian optimizer (Manopt). In this instance, $\mathcal{L}_1$--SVD is slightly faster (${\approx}\,5.17$\,ms vs.\ $12.51$\,ms), though the difference is minor.

\begin{table}[h]
  \centering
  \caption{Performance comparison of projection methods}
\label{tab:projection_performance}
  \small
  \begin{tabular}{@{}lrrr@{}}
    \toprule
    \textbf{Method} & \textbf{Time (s)} & \textbf{Vol.\footnotemark} & \textbf{Ratio} \\
    \midrule
    Model (truth) & N/A & $1.68{\times}10^{-2}$ & $1.00$ \\
    \midrule
    Max.\ Rotation & $1.25{\times}10^{-2}$ & $1.23{\times}10^{-1}$ & $7.33$ \\
    L1-SVD & $5.17{\times}10^{-3}$ & $1.23{\times}10^{-1}$ & $7.34$ \\
    Manopt & $113.2$ & $1.23{\times}10^{-1}$ & $7.36$ \\
    \bottomrule
  \end{tabular}
\end{table}
\footnotetext{%
Exact 5D volume evaluation at $t{=}5$ requires
$>$9000\,GB memory; we report 3D projected volumes.
The 5D values scale approximately as
$V_{\text{5D}}\!\approx\! V_{\text{3D}}^{5/3}$.}


Our system is linear time-invariant (LTI), so the projection matrix $P$ is computed once offline before the reachability analysis. Extending the same ``orthogonally align then reduce'' idea to linear time-varying (LTV) systems would require updating $P$ at each step, since the uncertainty structure changes over time. This is a nontrivial extension, particularly in the data-driven setting. Both $\mathcal{L}_1$--SVD and maximal rotation can recompute $P$ in milliseconds per step if needed, which makes them practical candidates for such future extensions.

We additionally consider a piecewise affine system with a switching surface at $x_1 = 0$ subject to bounded noise~\cite{xie2025data}.
As shown in Figure~\ref{fig:hybrid-nmz-mz}, the sets $\hat{\mathcal{R}}^{\mathrm{mr}}_k$ (red) obtained by the
maximal-rotation method yield a tighter outer approximation than the sets $\hat{\mathcal{R}}_k$
(black solid line) obtained directly from data, while remaining close to the model-based reachable sets
$\mathcal{R}_k$ (blue). Hybrid systems place stricter demands on the tightness of reachable sets:
errors near the switching boundary (at $x_1 = 0$ in Figure~\ref{fig:hybrid-nmz-mz}) can lead to exponentially accumulated propagation errors that severely degrade subsequent reachable-set computations~\cite{xie2025data}.

\begin{figure}
\centering
\includegraphics[width=0.95\columnwidth]{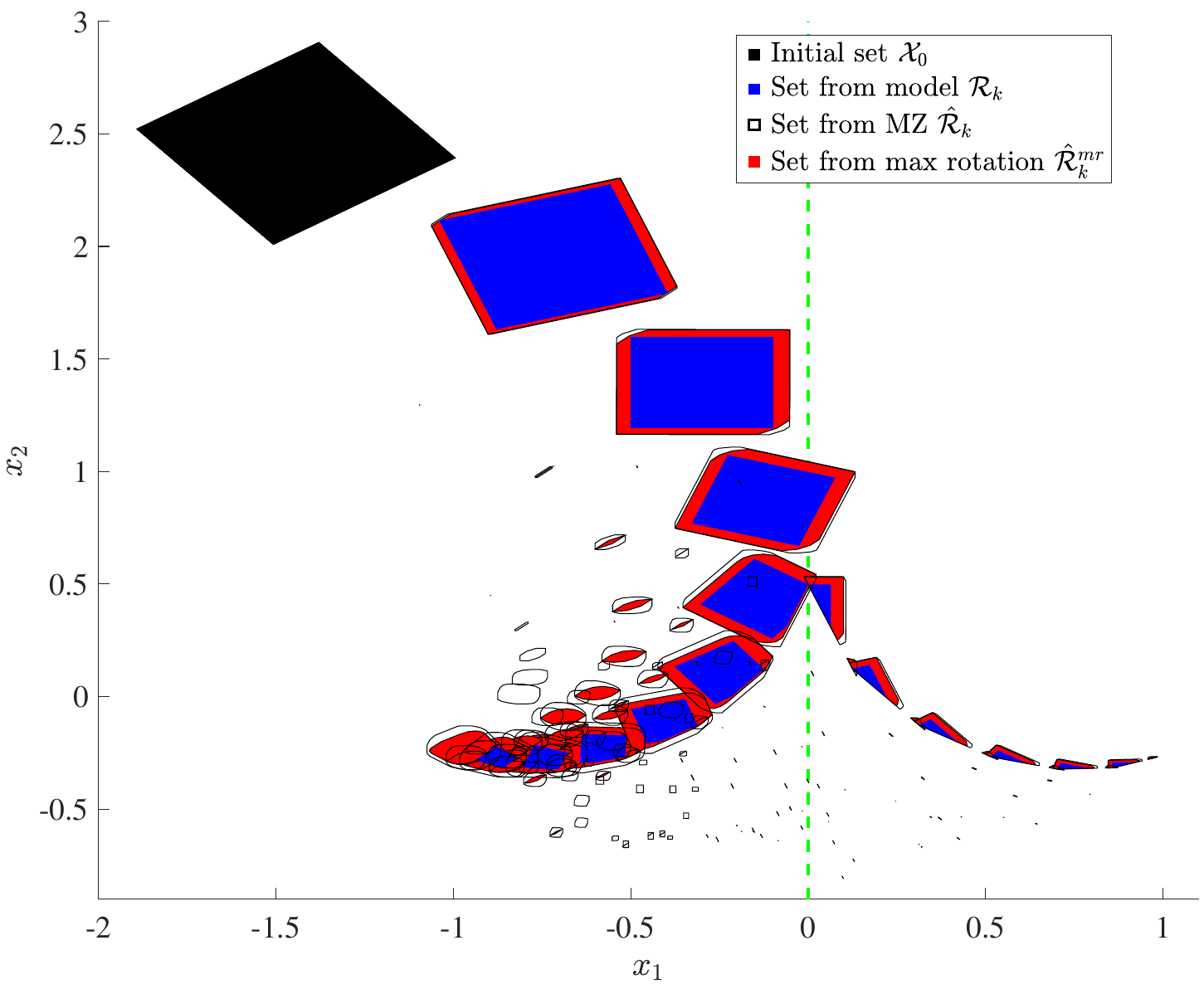}
\caption{\textbf{Piecewise affine system (switching at $x_1=0$).} Initial set $\mathcal{X}_0$ (black); sets from model $\mathcal{R}_k$ (blue); sets from matrix zonotope (MZ) $\hat{\mathcal{R}}_k$ (black solid line); sets from maximal rotation $\hat{\mathcal{R}}^{\mathrm{mr}}_k$ (red) which yield a tighter outer approximation.}
\label{fig:hybrid-nmz-mz}
\end{figure}

\section{Conclusion}\label{conclu}
This paper addressed the coordinate-dependent conservatism introduced by fixed-order generator reduction in data-driven reachability analysis. By applying a suitably chosen orthogonal transformation before Girard reduction, the proposed method produces outer approximations with substantially smaller volume while maintaining a comparable generator count. The rank-one structure of data-driven matrix zonotope generators was exploited to construct effective orthogonal matrices at low computational cost, demonstrating that structured uncertainty can be \emph{aligned} rather than merely over-approximated.

Several directions merit further investigation. First, sharper perturbation bounds such as the Cai--Zhang theorem may yield tighter guarantees for subspace alignment. Second, extending the framework to constrained matrix zonotopes and nonlinear systems would broaden its applicability to hybrid and safety-critical settings. Third, reducing the effective number of generators~$\gamma_M$ through low-rank noise approximation could mitigate the growth factor in~\eqref{eq:gen-growth-exact}. Finally, the interplay between orthogonal transformations and discrete transitions in hybrid systems presents opportunities for further conservatism reduction.

\balance
\bibliography{BibTex_2025}
\bibliographystyle{ACM-Reference-Format}

\end{document}